\documentclass[10pt,a4paper,hyperref]{elsarticle}

\usepackage{algorithm,algorithmic}
\usepackage{latexsym,amssymb,amsmath,amsfonts,amsthm,enumerate,tikz,graphicx}
 
\usetikzlibrary{calc}

\usepackage{hyperref}
\hypersetup{
	colorlinks,
	linktoc=all,
	pdfstartview={FitH}
}

\newcommand{\cl}[1]{{\rm col}({#1})}
\newcommand{\mcl}[1]{{\rm mcol}({#1})}
\newcommand{\dfs}[2]{{\rm dfs}({#2};{#1})}
\newcommand{\rp}[1]{{\rm rep}({#1})}

\newtheorem{theorem}{Theorem}
\newtheorem{corollary}{Corollary}
\newtheorem{definition}{Definition}
\newtheorem{lemma}{Lemma}

\begin{document}

\title{Improved Approximation for Maximum Edge Colouring Problem}

\author{L Sunil Chandran}
\ead{sunil@iisc.ac.in}
\address{Department of Computer Science and Automation, Indian Institute of Science, Bangalore}

\author{Abhiruk Lahiri}
\ead{abhiruk@iisc.ac.in}
\address{Department of Computer Science and Automation, Indian Institute of Science, Bangalore}

\author{Nitin Singh}
\ead{nitisin1@in.ibm.com}
\address{IBM India Research Lab, Bangalore}

\begin{abstract}
The \emph{anti-Ramsey number}, $ar(G, H)$ is the minimum integer $k$ such that in any edge colouring of $G$ with $k$ colours there is a \emph{rainbow subgraph} isomorphic to $H$, namely, a copy of $H$ with each of its edges assigned a different colour. The notion was introduced by Erd{\"{o}}s and Simonovits in 1973. Since then the parameter has been studied extensively. The case when $H$ is a star graph was considered by several graph theorists from the combinatorial point of view. Recently this case received the attention of researchers from the algorithm community because of its applications in interface modelling of wireless networks. To the algorithm community, the problem is known as maximum edge $q$-colouring problem: Find a coloring of the edges of $G$, maximizing the number of colors satisfying the constraint that each vertex spans at most $q$ colors on its incident edges. It is easy to see that the maximum value of the above optimization problem equals
$ar(G, K_{1,{q+1}})-1$.

In this paper we study the maximum edge $2$-coloring problem from the approximation algorithm point of view.  The case $q=2$ is particularly interesting due to its application in real life problems. Algorithmically, this problem is known to be NP-hard for $q\ge2$. For the case of $q=2$, it is also known that no polynomial time algorithm can approximate to a factor less than $3/2$ assuming the unique games conjecture. Feng et al. showed a $2$-approximation algorithm for this problem. Later Adamaszek and Popa presented a $5/3$-approximation algorithm with the  additional assumption that the input graph has a perfect matching. Note that the obvious but the only known algorithm issues different colours to the edges of a maximum matching (say $M$)  and different colours to the  connected components of  $G \setminus M$ . In this article, we give a new analysis of the aforementioned  algorithm leading to an  improved approximation bound for triangle-free graphs with perfect matching. We also show a new lower bound when the input graph is triangle-free. The contribution of the paper is a completely new, deeper and closer analysis of how the optimum achieves higher number of colors than the matching based algorithm, mentioned above. 
\end{abstract}

\journal{arXiv.org}

\maketitle

\section{Introduction}
A $k$-\emph{edge colouring} of a graph is a function $f \colon E(G) \rightarrow [k]$. Note that $f$ doesn't need to be  a proper colouring of the edges, i.e., edges incident to the same  vertex may receive the same colour. A subgraph $H$ of $G$ is called a \emph{rainbow subgraph} (\emph{heterochromatic subgraph}) with respect to a $k$-edge colouring $f$ if all the edges of $H$ are coloured distinctly. For a pair of graphs $G$ and $H$ the \emph{anti-Ramsey number}, $ar(G, H)$, denotes the minimum number of colours $k$ such that in \emph {any}  $k$-edge colouring of $G$ there exists at least one subgraph isomorphic to $H$ which is a rainbow subgraph. Equivalently if $k'$ is the maximum possible number of colours in an edge colouring $f$  of $G$  such that  there exists no rainbow subgraph isomorphic to $H$ with respect to $f$ then $ar(G,H) = k'+1$. We call the first parameter of $ar(G,H)$,  $G$ as the \emph{input graph} and the second parameter $H$ as the \emph{pattern graph}. 

The notion, anti-Ramsey number, was introduced by Erd{\"{o}}s and Simonovits in 1973 \cite{Erdos}. Most of the initial research on this topic focused on complete graphs ($K_n$) as the input graph and pattern graphs that possesses certain nice structure, for example, path, cycle, complete graph etc. The exact expression of $ar(K_n, P_k)$, when the pattern graph is a path of length $k$ ($P_k$), was reported in the article written by Simonovits and S{\'o}s \cite{Simono}. Whereas the simple case of the pattern graph is a cycle of length $k$ ($C_k$) took years to get solved completely. It was proved by Erd{\"{o}}s, Simonovits and S{\'{o}}s that $ar(K_n, C_3) = n-1$ \cite{Erdos}. In the same paper it was conjectured that $ar(K_n, C_n) = (\frac{k-2}{2} + \frac{1}{k-1})n + O(1)$ for $k\geq 4$. The conjecture was verified affirmatively for the case $k =4$ by Alon \cite{Alon83}. Later it was studied by Jiang and West \cite{JiangW03}. Almost thirty years after it was conjectured, Montellano-Ballesteros and Neumann-Lara reported proof of the statement in 2005 \cite{MontelN05}. A lower bound considering the pattern graph as the clique of size $n-1$ ($K_{n-1}$)  was reported in \cite{ManousSTV96}. Schiermeyer and Montellano-Ballesteros together with Neumann-Lara independently reported the exact value of $ar(K_n, K_r)$  \cite{Schier04, MontelN02}. In the same article Schiermeyer also studied the case when pattern graph is a matching. Haas and Young later studied the case when pattern graph is a perfect matching \cite{HaasY12}. A tighter bound on matching was reported in the article by Fujita et al. \cite{FujitaKSS09}. The article by Jiang and West reported bounds on $ar(G, H)$ when $H$ is a tree \cite{JiangW04}. Jiang also derived an upper bound on $ar(G,H)$ when $H$ is a complete subdivided graph relating the parameter $ar(G,H)$ with \emph{Tur{\'{a}}n number}, that is the maximum cardinality of edges of an $n$-vertex graph that does not contain a subgraph $H$ \cite{JiangJCT02}. 

The study of anti-Ramsey number was not entirely restricted to the case when the input graph is a complete graph. Axenovich et al. studied the case when the input graph is a complete bipartite graph \cite{AxenoJK04}. A $t$-round variant of anti-Ramsey number was introduced and studied in \cite{BlokhuFGR01}.

In this paper, we consider the pattern graph as the claw graph, i.e. the star graph with exactly 3 leaves, denoted by $K_{1,3}$. Study of anti-Ramsey number where the pattern graph is the claw or more generally the star graph was initiated in the work of Manoussakis et al.  \cite{ManousSTV96}. The bound was later improved in \cite{Jiang02}. In the same article exact value of the bipartite variant of the problem $ar(K_{n,n}, K_{1,q})$ was also reported. Gorgol and Lazuka computed the exact value of $ar(G,H)$ when $H$ is  $K_{1,4}$ with an edge added  to it \cite{GorgolL10}. Montellano-Ballesteros relaxed the condition on input graph and considered any graph as input in their study \cite{Montel06}. The study of anti-Ramsey number with claw graph as pattern graph was revisited recently due to its application in modelling channel assignment in a network equipped with a multi-channel wireless interface \cite{Raniwala}. They introduced the problem as \emph {maximum  edge $q$-colouring problem}, thus initiating the exploration of the  algorithmic aspects of this parameter, $ar(G, K_{1,t})$.

For a graph $G$, an edge $q$-colouring of $G$ is an assignment of colours to edges of $G$ such that no more than $q$ distinct colours are incident at any vertex. An optimal edge $q$-colouring is one which uses the maximum number of colours. It is easily seen that the number of colours in maximum edge $q$-colouring of $G$ is $ar(G, K_{1, q+1})-1$.

In \cite{AdamP}, it was reported that the problem is $NP$-hard for every $q \geq 2$. Moreover, they showed that it is hard to approximate within a factor of $(1+ 1/q)$ for every $q \geq 2$, assuming the \emph{unique games conjecture}. A simple $2$-factor algorithm for maximum $2$-colouring problem was reported in \cite{Feng}. A description of the algorithm is provided in Algorithm \ref{algo01}. Henceforth we refer this algorithm as the \emph{matching based algorithm}. In a recent article \cite{AdamP}, authors reported a  $5/3$ approximation factor for the same algorithm assuming that the input graphs have a perfect matching. Approximation bounds for the matching based algorithm when the input graph has certain degree constraints were reported in \cite{Chandran}. A fixed-parameter tractable algorithm was reported for the case $q=2$ in \cite{Goyal}.

In the present article, our focus is on the case when $q=2$, and when the graph $G$ has a perfect matching. It is worth mentioning here, although Montellano-Ballesteros reported bounds on $ar(G, K_{1,q})$, their expression is not enough to draw any inference in this particular scenario. Their technique is useful for deriving bounds when the input graph has certain regular structures such as complete graph, complete $t$-partite graph, hypercube etc.
\begin{algorithm}
\caption{Matching based algorithm for edge $2$-colouring}
\label{algo01}
\begin{algorithmic}[1]
\REQUIRE A graph $G$ which has perfect matching.
\STATE Compute the maximum matching $M$ of $G$.  
\STATE Assign distinct colour to each edge of $M$.  
\STATE Assign a new colour to each component of $G \setminus M$.
\RETURN The colour assigned graph.
\end{algorithmic}
\end{algorithm}

\section{Key Notation and Main Result}
\label{sec:keynotation}
Throughout this article (except possibly the last section), we consider $G$ to be a graph which has a perfect matching $M$. We use $C_1,\ldots, C_h$ to denote the components of $G\backslash M$, where $h$ is the number of such components.

Let $\mathcal{C}$ be an optimal edge $2$-colouring of $G$ using colours $[c] := \{1,\ldots, c\}$. Let $\mathcal{C}_M$ and $\mathcal{C}_N$ denote the colours used in the matching $M$, and those not used in the matching $M$ respectively. Clearly $\mathcal{C}_M\uplus \mathcal{C}_N = [c]$. We call colours in $\mathcal{C}_M$ as \emph{matching} colours and colours in $\mathcal{C}_N$ as \emph{non-matching} colours. For an edge $e$, we denote the colour assigned to $e$ in colouring $\mathcal{C}$ by $\cl{e}$. For a vertex $u$, the colour assigned to the matching edge at $u$ is denoted by $\mcl{u}$.

For a colour $i\in [c]$, let $G[i]$ denote the subgraph spanned by the edges coloured $i$. We note that, in an optimal colouring $\mathcal{C}$,  $G[i]$ is connected for all colours $i\in [c]$, since otherwise we can increase the number of colours. For a non-matching colour $i$, $G[i]$ is a subgraph of a unique component $C_j$. For convenience, we often refer to a connected subgraph spanned by edges of a colour as \emph{colour component}. Subgraphs corresponding to matching colours are called \emph{matching colour components}, while those corresponding to non matching colours are called \emph{non matching colour components}. With the notation as discussed, the following are the main contributions of this paper:
\begin{theorem}
\label{thm:main}
Let $G$ be a graph having a perfect matching $M$. Let $OPT$ denote the number of  colors in an optimal max $2$-edge coloring of $G$. Then:
\begin{enumerate}[{\rm (a)}]
\item $OPT\leq \frac{5}{3}(|M|+h)$ \emph{\cite{AdamP}}.
\item $OPT\leq \frac{8}{5}(|M|+h)$ when $G$ is additionally triangle-free.
\end{enumerate}
\end{theorem}
 
\begin{corollary}
\label{cor:main}
Algorithm \ref{algo01} guarantees an approximation factor $5/3$ for graphs with perfect matching and an approximation factor of $8/5$ for triangle free graphs with perfect matching.
\end{corollary}

\section{Overview}
We start with an overview of some structural observations about an optimal colouring which help us establish the approximation factor. Let $C$ be one of the components of $G\backslash M$ and let $H_i = G[i]$ and $H_j = G[j]$ be connected subgraphs of $C$ spanned by some non-matching colours $i$ and $j$. Note that $V(H_i) \cap V(H_j) = \phi$. Our next lemma shows that any $H_i$-$H_j$ path contains two distinct vertices $u$ and $v$ such that the matching edges incident at $u$ and $v$ have the same colour. Let us call $(u,v)$ as the \emph{colour repetition pair}.

\begin{lemma}\label{lem:pathlemma}
Let $u_0u_1\cdots u_k$ be a path in $G\backslash M$ such that $\mcl{u_0}=\cl{u_0u_1}$ and $\mcl{u_k}=\cl{u_{k-1}u_k}$. Then there exist $0\leq i < j\leq k$ such that $\mcl{u_i}=\mcl{u_j}$.
\end{lemma}

\begin{proof}
For $k=1$ the lemma is obvious. Let $k\geq 2$. Let $j$ be maximal so that the path $u_0\cdots u_j$ is monochromatic. If $\mcl{u_0}=\mcl{u_j}$ then $0$ and $j$ satisfy the assertion of the lemma. Otherwise, clearly $j < k$ and $\cl{u_ju_{j+1}}=\mcl{u_j}$. Hence, the lemma follows by applying induction on the path $u_j\cdots u_k$.
\end{proof}  

We observe that any $H_i$-$H_j$ path satisfies the conditions in Lemma \ref{lem:pathlemma} at its end-points. This is because if $(u, v)$ is an edge coming out of the non-matching colour component $H_i$, with $u \in H_i$ and $v \notin H_i$, $\cl{uv}$ has to be the same as $\mcl{u}$. Now suppose $C$ has $k$ such non-matching colours. Then we can find at least $k-1$ paths connecting all of these colour components. If these paths are all disjoint (as in Figure \ref{fig:disjoint}), we would have $k-1$ colour repetition pairs by applying the Lemma \ref{lem:pathlemma} on each of these paths. Intuitively, a lot of such pairs should imply ``repetition'' of colours among matching edges, and hence help us bound the number of distinct matching colours. In later sections, we try to quantify the repetition implied by these colour repetition pairs. Our focus here and in the next section is to exhibit a large number of such pairs.

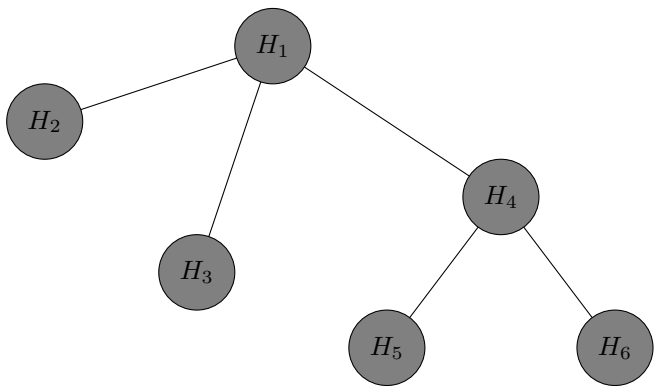
\begin{figure}[ht]
\centering
\begin{tikzpicture}
\coordinate (A) at (0,0);
\coordinate (B) at (-3,-1);
\coordinate (C) at (-1,-3);
\coordinate (D) at (3,-2);
\coordinate (E) at (1.5, -4);
\coordinate (F) at (4.5, -4);

\draw (A) -- (B) (A) -- (C) (A) -- (D) (D) -- (E) (D) -- (F);
\draw[fill=gray] (A) circle(0.5cm);
\draw[fill=gray] (B) circle(0.5cm);
\draw[fill=gray] (C) circle(0.5cm);
\draw[fill=gray] (D) circle(0.5cm);
\draw[fill=gray] (E) circle(0.5cm);
\draw[fill=gray] (F) circle(0.5cm);

\draw (A) node  {$H_1$};
\draw (B) node  {$H_2$};
\draw (C) node  {$H_3$};
\draw (D) node  {$H_4$};
\draw (E) node  {$H_5$};
\draw (F) node  {$H_6$};
\end{tikzpicture}
\caption{Circles denote subgraphs formed by non-matching colours. The figure shows the case when all the subgraphs can be connected by disjoint paths between them}
\label{fig:disjoint}
\end{figure}

Evidently, the non-matching components may not be nicely connected as in Figure \ref{fig:disjoint}. In the Section \ref{sec:colrep}, we show that one can still find $k-1$ distinct colour repetition pairs from a component containing $k$ non-matching colours. In the Section \ref{sec:repetcontent}, we estimate quantitatively the repetition in matching colours. The result is given in Lemma \ref{lem:repetitioncontent}. We use it to prove our main theorem in the same section. 

\section{Colour Repetition Pairs}
\label{sec:colrep}
We begin by generalizing the Lemma \ref{lem:pathlemma} to rooted trees, where the endpoint conditions in the Lemma \ref{lem:pathlemma} are satisfied at the root and the leaves of the tree. For a tree $T$, let $r(T)$ and $l(T)$ denote the root and the set of leaves of $T$ respectively. Further, we assume that $T$ has a depth first ordering, where $\dfs{T}{v}$ denotes the index of vertex $v \in V(T)$ in the ordering. We assume that when $v$ is a descendant of $u$ in $T$, we have $\dfs{T}{u} < \dfs{T}{v}$. Thus the root has minimum index. We use the depth first ordering to define an ordering $\preceq$ on vertices of $T$, where $u\preceq v$ if $\dfs{T}{u} \geq \dfs{T}{v}$. We will use $u \prec v$ to denote $u\preceq v$ with $u \neq v$. Note that this ordering contains the usual hierarchical ordering with root as the maximum element. This is same as the post order traversal of trees found in the literature. We are now in a position to state our next lemma. 

\begin{lemma}\label{lem:genrep}
Let $T$, $|V(T)| > 1$, be a rooted tree in $G\backslash M$ with $r := r(T)$ and $L := l(T)$. Further for $u\in L\cup \{r\}$, the colour of the edges in $T$ incident at $u$ is the same as the matching colour at $u$. Then we have a set $P := \{(u_i,v_i): i=1,\ldots, |L|\}$ of pairs of vertices of $T$ such that:
\begin{enumerate}[{\rm (a)}]
\item For $i\neq j$, we have $u_i\neq u_j$.
\item $u_i\prec v_i$ for all $1\leq i\leq |L|$.
\item $\mcl{u_i} = \mcl{v_i}$ for all $1\leq i\leq |L|$.
\item The path $u_iTv_i$ is monochromatic, with edges coloured $\mcl{u_i}$ for all $1\leq i\leq |L|$.
\item For an internal vertex $z$ in the path $u_iTv_i$, we have $\mcl{z}\neq \mcl{u_i}$ for $1\leq i\leq |L|$.
\item For $i\neq j$ and $u_iv_i, u_jv_j\in M$, the paths $u_iTv_i$ and $u_jTv_j$ do not share an internal vertex.
\end{enumerate}
\end{lemma}
We call the tuples in $P$ as the repetition pairs. Statements (b) and (c) can be seen as generalization of the notion of colour repetition pairs for paths in a tree. As in the proof of Lemma \ref{lem:pathlemma}, for a repetition pair $(u,v)$, $v$ was the vertex with the same incident matching colour as $u$ but with an index to the right of $v$ closer to root. Similarly in a tree, we have $v_i$'s ``closer'' to the root in depth first ordering than $u_i$'s. The statements (d)-(f) establish further structural properties of these repetition pairs that are needed subsequently. 
 
\begin{proof}[Proof of Lemma \ref{lem:genrep}]
We first prove the Lemma for the case when $T$ is monochromatic. Let $a$ be the colour of all the edges of $T$. From the colouring constraints that $T$ satisfies at the root and the leaves, it follows that $\mcl{u}=a$ for $u\in L\cup \{r\}$. For $u\in L$, define $f(u)$ to be the closest vertex to $u$ on the path $uTr$ such that $\mcl{f(u)}=\mcl{u}=a$. Observe that $f(u)$ exists for each $u$ as the root $r$ satisfies $\mcl{r}=a$. Let $P := \{(u,f(u)): u\in L\}$. Then $P$ is trivially seen to satisfy statements (a)-(e). To prove (f), we assume the contrapositive, i.e, there exist $u\neq v$ such that $uf(u)\in M$, $vf(v)\in M$ and the paths $u$-$f(u)$ and $v$-$f(v)$ in $T$ have a common internal vertex. Let $w$ be the common internal vertex. Then $u$ and $v$ are descendants of $w$, while $f(u), f(v)$ are ancestors of $w$. But $f(u)$ and $f(v)$ are the first vertices on $uTr$ and $vTr$ respectively with matching colour as $a$. Since the paths $uTr$ and $vTr$ share the sub-path $wTr$, it must be that $f(u)=f(v)$. But then their matching partners must also be the same, i.e., $u=v$, which contradicts the assumption that $u\neq v$. This contradiction proves (f).

Next consider the case when $T$ is not monochromatic. We proceed by induction. Let $v$ be a vertex of minimum height, such that $v$ witnesses edges of two colours in $T$. Let $a$ and $b$ be the colours incident at $v$. Without loss of generality, let $a$ be the colour of the matching edge at $v$, i.e., $\mcl{v}=a$. Let $T_v$ be the subtree of $T$ rooted at $v$. Let $T_a$ and $T_b$ be the subtrees of $T_v$ consisting of edges of colours $a$ and $b$ respectively. We consider three cases, viz:\medskip

\noindent\emph{$T_a$ is non-empty and $T_b$ is also non-empty}: Let $L_a$ denote the leaves in $T_a$. Now $T_a$ is monochromatic and satisfies the end point constraints required by the lemma. Thus we have the set $P_a$ with $|P_a|=|L_a|$ of pairs of vertices of $T_a$ satisfying the statements (a)-(f) because of the previous case. Let $T'$ be the tree obtained from $T$ by deleting all descendants of $v$ in $T_a$. Let $L'$ be the set of leaves in $T'$. Notice that $L_a\uplus L' = L$. Now, by induction hypothesis, we have a set $P'$ of pairs of vertices of $T'$ satisfying (a)-(f). Thus both $P_a$ and $P'$ satisfy the statements (a)-(f) with respect to the trees $T_a$ and $T'$. We claim that $P=P_a\cup P'$ satisfies (a)-(f) for the tree $T$. It is trivially seen that (a)-(e) hold. To see that (f) also holds we observe that only vertex that is possibly shared between $T_a$ and $T'$ is $u$, and $u$ is not an internal vertex of any $u'T_av'$ path in $T_a$ where $(u',v')\in P_a$.\medskip

\noindent\emph{$T_a$ is non-empty and $T_b$ is empty}: Let $w_0 = v$ and let $w_0, w_1, \dots w_j$ be the subpath of $vTr$ such that $w_j$ has degree more than $2$ in $T$ and height of $w_j$ is minimum among all the vertices with this property. Let the subtree rooted at $w_j$ which contains $v$ be $T_1$. Let $T'$ be the tree obtained from $T$ by removing all descendants of $w_j$ in $T_1$ (possibly $T'=\emptyset$). Since $w_j$ is not an internal vertex of $T'$ we see that $T'$ satisfies the end point constraints required by the lemma. Let leaves of $T'$ be $L'$. Observe that $L = L' \uplus L_a$. Now, by induction hypothesis, we have a set $P'$ of pairs of vertices of $T'$ satisfying (a)-(f) and a set $P_a$ of pairs of vertices of $T_a$ satisfying (a)-(f). We claim that $P=P_a\cup P'$ satisfies (a)-(f) for the tree $T$. It is trivially seen that (a)-(e) hold. To see that (f) also holds we observe that only vertex that is possibly shared between $T_a$ and $T'$ is $w_j$, and $w_j$ is not an internal vertex of any $u'T_av'$ path in $T_a$ where $(u',v')\in P_a$. \medskip  

\noindent\emph{$T_a$ is empty}: Clearly, $T_b$ must be non-empty in this case. Let $L_b$ denote the leaves in $T_b$. Note that the colour of incident edges at $v$ in $T_b$ is not the same as the matching colour at $v$, and hence we cannot define the pairs $(u,f(u))$ as in the monochromatic case. Let $w$ be the leaf such that $\dfs{T}{w}$ is maximum in $L_b$. For each $u\in L_b, u\neq w$ define $g(u)$ as the closest vertex to $u$ on the path $uT_bw$ such that $\mcl{g(u)}=b$. Again, observe that $g(u)$ exists for all $u\in L_b\backslash \{w\}$ because $\mcl{u} = \mcl{w} = b$. Define $P_b = \{(u, g(u)): u\in L_b\backslash \{w\}\}$. Statements (a)-(e) follow for $P_b$ (w.r.t the tree $T_b$) almost by definition of function $g$. Statement (f) can be proven in the following way. There is a unique path between any two pairs of vertices. Then at some point, the two paths have a vertex in common. Thus, the pairs $P_b$ satisfy (a)-(f) for the tree $T_b$, except that the number of pairs is one short of the number of leaves in $T_b$. We recover the deficit in the remaining tree. As before, let $T'$ be the tree obtained from $T$ by removing the descendants of $v$ in $T_b$. Let $L'$ denote the leaves in $T'$. Since $T_a$ was empty and $v$ witnesses two colours, we have $v\in L'$ and it is incident with an $a$-coloured edge in $T'$. By induction hypothesis, we have pairs $P'$ of vertices of $T'$ with $|P'|=|L'|$ satisfying statements (a)-(f). Again, we claim that $P=P_b\cup P'$ satisfy the requirements of the lemma for $T$. Statements (a)-(e) are easily verified. For (f), we note that $v$ is a leaf of $T'$ and hence is not an internal vertex of any path $u'T'v'$ in $T'$ for $(u',v')\in P'$. Thus, (f) also holds for the pairs $P$. Finally, we note that $|P|=|L|$, as we compensate the loss of leaf $w$ in $T_b$ with an extra leaf $v$ in $T'$ since $|P_b| = |L_b| - 1$ and $|P'| = |L| - |L_b| + 1$.
\end{proof}

Lemma \ref{lem:genrep} can be easily extended to forest consisting of rooted trees. Let $F$ be a forest containing rooted trees. Let $r(F)$ denote the set of roots of trees in the forest, and $l(F)$ denote the set of leaves in the forest $F$. We define a partial order $\preceq$ on the forest which restricts to the (total) order $\preceq_T$ (as defined earlier) on each component tree $T$. If two vertices belong to different component trees, they are incomparable under $\preceq$. As before $u \prec v$ denotes $u \preceq v$, but $u \neq v$. We now state the extension of Lemma \ref{lem:genrep} to forests. 

\begin{lemma}\label{lem:genrepforest}
Let $F$ be a forest in $G\backslash M$ with $R := r(F)$ and $L := l(F)$.
Suppose that for each $u\in R\cup L$, the colour of edges in $F$ incident at
$u$ is the same as the matching colour at $u$. Then, there exists set $P = \{(u_i,v_i): i = 1,\ldots,
|L|\}$ of pairs of vertices of $F$ satisfying:
\begin{enumerate}[{\rm (a)}]
\item For $i\neq j$, we have $u_i\neq u_j$.
\item $u_i\prec v_i$ for all $1\leq i\leq |L|$. In particular, the path
$u_iFv_i$ exists for all $1\leq i\leq |L|$.
\item $\mcl{u_i} = \mcl{v_i}$ for all $1\leq i\leq |L|$.
\item The path $u_iFv_i$ is monochromatic, with edges coloured $\mcl{u_i}$ for all $1\leq i\leq |L|$.
\item For an internal vertex $z$ in the path $u_iFv_i$, we have $\mcl{z}\neq \mcl{u_i}$ for $1\leq i\leq |L|$.
\item For $i\neq j$ and $u_iv_i, u_jv_j\in M$, the paths $u_iFv_i$ and $u_jFv_j$ do not share an internal vertex.
\end{enumerate}
\end{lemma}

\begin{proof}
The proof follows by taking union of pairs satisfying Lemma \ref{lem:genrep} for each component tree in the forest.
\end{proof}
  
Let us return to the question of finding $k-1$ color repetition pairs in a component $C \in G \setminus M$ containing $k$ non-matching color components. We could use Lemma \ref{lem:pathlemma}, when the non-matching components can be connected using pairwise disjoint paths as in Figure~\ref{fig:disjoint}. Lemma \ref{lem:genrepforest} allows us to exhibit $k-1$ repetition pairs, as long as we have a non-matching color component $H^*$ from which we can reach all other $k-1$ non-matching color components without passing through other colour components. Figure \ref{fig:treelike} illustrates such an arrangement, and the
corresponding forest $F$ to which we can apply Lemma \ref{lem:genrepforest}. Note that root of both trees in the forest is in $H^*$. The $H^\ast H_i$ path doesn't intersect with any other colour component other than $H^\ast$ and $H_i$. For each $i\neq j$, we can construct a $H_i$-$H_j$ path avoiding vertices of other colour component $H_l$ where $l\not\in \{i,j\}$.

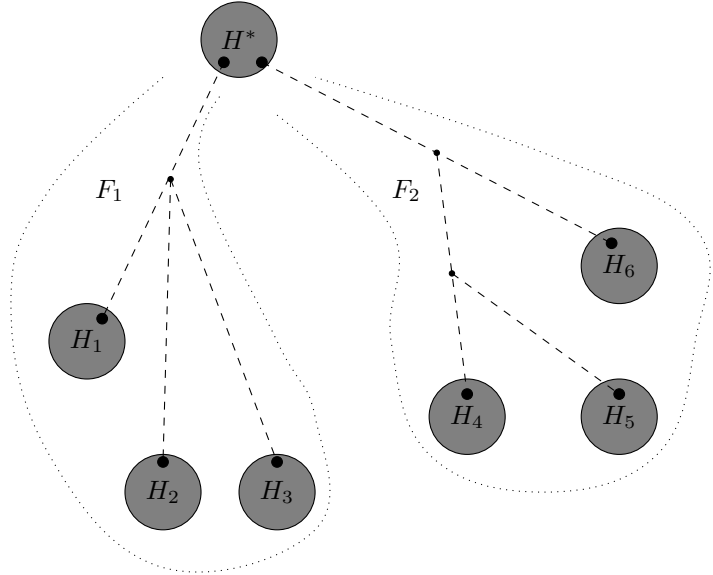
\begin{figure}[ht]
\centering
\begin{tikzpicture}
\coordinate (A) at (0,0);
\coordinate (B) at (-2, -4);
\coordinate (C) at (-1, -6);
\coordinate (D) at (0.5, -6);
\coordinate (E) at (5, -3);
\coordinate (F) at (3, -5);
\coordinate (G) at (5, -5);

\coordinate (L1) at (-1.7,-2);
\coordinate (L2) at (2.2,-2);

\coordinate (A1) at (-0.2,-0.3);
\coordinate (A2) at (0.3, -0.3);
\path (B) -- ++(0.2,0.3) coordinate (B1);
\path (C) -- ++(0.0, 0.4) coordinate (C1);
\path (D) -- ++(0.0, 0.4) coordinate (D1);
\path (E) -- ++(-0.1,0.3) coordinate (E1);
\path (F) -- ++(0.0, 0.3) coordinate (F1);
\path (G) -- ++(0.0, 0.3) coordinate (G1);

\draw[fill = gray] (A) circle(0.5cm);
\draw[fill = gray] (B) circle(0.5cm);
\draw[fill = gray] (C) circle(0.5cm);
\draw[fill = gray] (D) circle(0.5cm);
\draw[fill = gray] (E) circle(0.5cm);
\draw[fill = gray] (F) circle(0.5cm);
\draw[fill = gray] (G) circle(0.5cm);

\draw[fill = black] (A1) circle(2pt);
\draw[fill = black] (A2) circle(2pt);
\draw[fill = black] (B1) circle(2pt);
\draw[fill = black] (C1) circle(2pt);
\draw[fill = black] (D1) circle(2pt);
\draw[fill = black] (E1) circle(2pt);
\draw[fill = black] (F1) circle(2pt);
\draw[fill = black] (G1) circle(2pt);

\coordinate (BC) at ($(A)!0.5!(B1)$);

\draw [dashed] (A1) -- (B1) (BC) -- (C1) (BC) -- (D1);
\draw[fill=black] (BC) circle(1pt);

\coordinate (AE) at ($(A2)!0.5!(E1)$);
\coordinate (AEF) at ($(AE)!0.5!(F1)$);

\draw [dashed] (A2) -- (E1) (AE) -- (F1) (AEF) -- (G1);

\draw[fill=black] (AE) circle(1pt);
\draw[fill=black] (AEF) circle(1pt);

\draw (A) node  {$H^*$};
\draw (B) node  {$H_1$};
\draw (C) node  {$H_2$};
\draw (D) node  {$H_3$};
\draw (E) node  {$H_6$};
\draw (F) node  {$H_4$};
\draw (G) node  {$H_5$};

\draw [dotted] plot[smooth, tension=.7] coordinates {(-1,-0.5) (-2.5,-2) (-3,-3.5) (-2.5,-5.5) (-1,-7) (1,-6.5) (1,-5) (0.5,-4) (-0.5,-1.8) (-0.25,-0.75)};
\draw [dotted] plot[smooth, tension=.7] coordinates {(0.5,-1) (2,-2.5) (2,-4) (2.5,-5.5) (4,-6)(5.5,-5.5) (6,-4)(6,-2.5)(4,-1.5) (1,-0.5)};

\draw (L1) node {$F_1$};
\draw (L2) node {$F_2$};
\end{tikzpicture}
\caption{Non-matching colour components which can be connected in a tree-like fashion from a component. The dashed lines denote paths in $G\backslash M$, where we have not explicitly marked all vertices outside the non-matching components. The requisite number of repetition pairs are obtained by applying Lemma \ref{lem:genrepforest} to the forest consisting of dashed trees.}
\label{fig:treelike}
\end{figure}

Our final technical effort in this section is to show that we can find $k-1$ repetition pairs in a component with $k$ non-matching colours, even if we do not have a non-matching colour component, from which all other non-matching components are accessible without passing through some vertex of an intermediate component. The situation is illustrated in Figure \ref{fig:noroot}.

\begin{figure}[ht]
\centering
\begin{tikzpicture}
\coordinate (A) at (0,0);
\coordinate (B) at (-2, -4);
\coordinate (C) at (-1, -6);
\coordinate (D) at (0.5, -6);
\coordinate (E) at (5, -3);
\coordinate (F) at (3, -5);
\coordinate (G) at (5, -5);

\coordinate (A1) at (-0.2,-0.3);
\coordinate (A2) at (0.3, -0.3);
\path (B) -- ++(0.2,0.3) coordinate (B1);
\path (C) -- ++(0.0, 0.4) coordinate (C1);
\path (D) -- ++(0.0, 0.4) coordinate (D1);
\path (E) -- ++(-0.1,0.3) coordinate (E1);
\path (F) -- ++(0.0, 0.3) coordinate (F1);
\path (G) -- ++(0.0, 0.3) coordinate (G1);

\draw[fill = gray] (A) circle(0.5cm);
\draw[fill = gray] (B) circle(0.5cm);
\draw[fill = gray] (C) circle(0.5cm);
\draw[fill = gray] (D) circle(0.5cm);
\draw[fill = gray] (E) circle(0.5cm);
\draw[fill = gray] (F) circle(0.5cm);
\draw[fill = gray] (G) circle(0.5cm);

\draw[fill = black] (A1) circle(2pt);
\draw[fill = black] (A2) circle(2pt);
\draw[fill = black] (B1) circle(2pt);
\draw[fill = black] (C1) circle(2pt);
\draw[fill = black] (D1) circle(2pt);
\draw[fill = black] (E1) circle(2pt);
\draw[fill = black] (F1) circle(2pt);
\draw[fill = black] (G1) circle(2pt);

\coordinate (BC) at ($(A)!0.5!(B1)$);

\draw [dashed] (A1) -- (B1) (BC) -- (C1) (BC) -- (D1);
\draw[fill=black] (BC) circle(1pt);

\coordinate (AE) at ($(A2)!0.5!(E1)$);
\coordinate (AEF) at ($(AE)!0.5!(F1)$);

\draw [dashed] (A2) -- (E1) (AE) -- (F1) (AEF) -- (G1);

\draw[fill=black] (AE) circle(1pt);
\draw[fill=black] (AEF) circle(1pt);

\path (C1) -- ++(0,-0.8) coordinate (C2);

\path (C2) -- ++(-1,-2) coordinate (H);
\path (C2) -- ++(1,-2) coordinate (I);

\path (G1) -- ++(-1,-2) coordinate (J);
\path (G1) -- ++(1,-2) coordinate (K);
\path (G1) -- ++(0,-0.6) coordinate (G2);

\path (H) -- ++(0,0.4) coordinate (H1);
\path (I) -- ++(0,0.4) coordinate (I1);
\path (J) -- ++(0,0.4) coordinate (J1);
\path (K) -- ++(0,0.4) coordinate (K1);

\draw[fill = gray] (H) circle(0.5cm);
\draw[fill = gray] (I) circle(0.5cm);
\draw[fill = gray] (J) circle(0.5cm);
\draw[fill = gray] (K) circle(0.5cm);

\draw[fill = black] (C2) circle(2pt);
\draw[fill = black] (H1) circle(2pt);
\draw[fill = black] (I1) circle(2pt);
\draw[fill = black] (J1) circle(2pt);
\draw[fill = black] (K1) circle(2pt);
\draw[fill = black] (G2) circle(2pt);

\draw [thick] (C2) -- (H1) (C2) -- (I1) (G2) -- (J1) (G2) -- (K1);

\draw (A) node {$H_1$};
\draw (B) node {$H_2$};
\draw (C) node {$H_3$};
\draw (D) node {$H_4$};
\draw (E) node {$H_5$};
\draw (F) node {$H_6$};
\draw (G) node {$H_7$};
\draw (H) node {$H_8$};
\draw (I) node {$H_9$};
\draw (J) node {$H_{10}$};
\draw (K) node {$H_{11}$};

\end{tikzpicture}
\caption{There is no component which can directly reach all other components. Choosing $H_1$ as root, the first level forest $F_1$ consists of two trees indicated by thin dashed lines. Choosing leaf components $H_3, H_7$ in $F_1$ as root components, we reach other components using forest $F_2$ denoted by thick lines.}
\label{fig:noroot}
\end{figure}
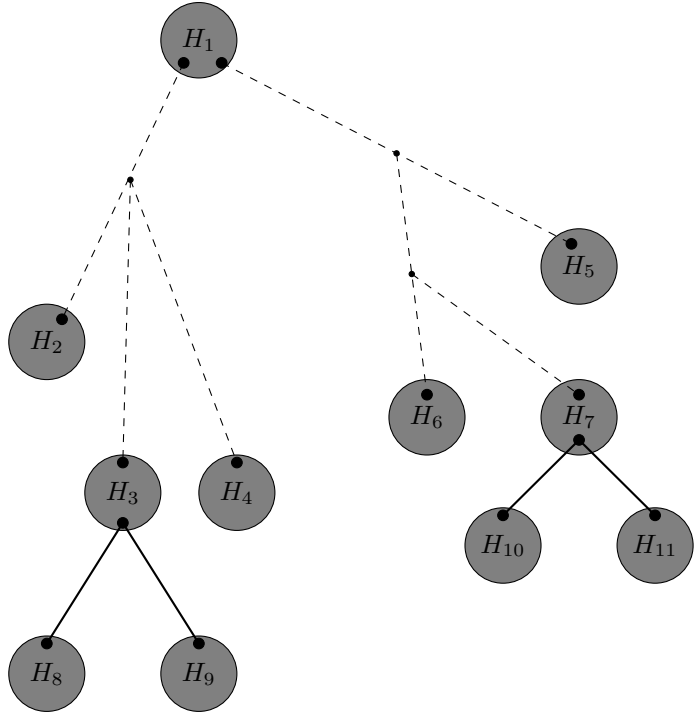

\subsection{Cascading sequence of forests}
\label{sec:cascading}
We call an ordered pair $(T_1, T_2)$ of rooted trees to be a \emph{cascading pair} if $T_1$ and $T_2$ are vertex disjoint or $r(T_2) \in l(T_1)$ and $|V(T_1) \cap V(T_2)| = 1$. We call $\mathcal{F}=\{F_1,\ldots, F_l\}$ a \emph{cascading sequence of forests} if  (i) $V(F_i) \cap V(F_j)=\emptyset$ for $|i -j| \geq 2$ and (ii) For, $T_i \in F_i$ and $T_{i+1} \in F_{i+1}$, the pair $(T_i, T_{i+1})$ is a cascading pair for $1 \leq i \leq l-1$.

Let $\mathcal{F}$ be a cascading sequence of forests. We use the notation $V(\mathcal{F}) := \bigcup_{F\in \mathcal{F}} V(F)$ to denote the vertices in the collection $\mathcal{F}$. We call $v\in V(\mathcal{F})$ to be an \emph{internal} vertex of $\mathcal{F}$ if $v$ is an internal vertex of some (at most one) forest $F\in \mathcal{F}$. The notation ${\rm Int}(\mathcal{F})$ will denote the set of internal vertices in the collection $\mathcal{F}$. 
We record the following easy observation as a lemma.

\begin{lemma}\label{lem:int_dis_forest_family}
Let $\mathcal{F}$ be a cascading sequence of forests. Then for any two distinct vertices $x,y$, there exists at most one forest $F\in \mathcal{F}$ such that $F$ contains the $x$-$y$ path.
\end{lemma}

\begin{definition}[Order on cascading sequence of forests]\label{defn:order}
Let $\mathcal{F}$ be an cascading sequence of forests. We define the partial order $\preceq_\mathcal{F}$ on $\bigcup_{F\in \mathcal{F}} V(F)$ as:  $x\preceq_\mathcal{F} y$ if and only if $x\preceq_F y$ for some (at-most one) $F\in \mathcal{F}$. Otherwise $x$ and $y$ are incomparable under $\preceq_\mathcal{F}$. Transitive closure of this relation is the order we consider. (By abuse of notation, we will use $\preceq_\mathcal{F}$ for transitive closure of this relation.)
\end{definition}

For a cascading sequence of forests $\mathcal{F}$ and vertices $x$ and $y$, the notation $x\mathcal{F}y$ denotes the path $xFy$ if there exists a forest (at-most one) $F\in \mathcal{F}$ such that $F$ contains $x$-$y$ path. We now state a version of Lemma \ref{lem:genrepforest} for a cascading sequence of forests.

\begin{lemma}\label{lem:genrepforest_collection}
Let $\mathcal{F}$ be a cascading sequence of forests in $G\backslash M$. Then, there exist pairs $P = \{(u_i,v_i): i = 1,\ldots, L\}$ of vertices in $V(F)$ where $L := \sum_{F\in \mathcal{F}} |l(F)|$,  satisfying:
\begin{enumerate}[{\rm (a)}]
\item For $1\leq i < j \leq L$, we have $u_i\neq u_j$.
\item $u_i\prec_\mathcal{F} v_i$ for all $1\leq i\leq L$. In particular for all $1\leq i\leq L$, the path $u_i\mathcal{F}v_i$ exists.
\item $\mcl{u_i} = \mcl{v_i}$ for all $1\leq i\leq L$.
\item The path $u_i\mathcal{F}v_i$ is monochromatic, with edges coloured $\mcl{u_i}$ for all $1\leq i\leq L$.
\item For an internal vertex $z$ in the path $u_i\mathcal{F}v_i$, we have $\mcl{z}\neq \mcl{u_i}$ for $1\leq i\leq L$.
\item For $i\neq j$ and $u_iv_i, u_jv_j\in M$, the paths $u_i\mathcal{F}v_i$ and $u_j\mathcal{F}v_j$ do not share an internal vertex.
\end{enumerate}
\end{lemma}

\begin{proof}
The proof follows by taking union of pairs satisfying Lemma \ref{lem:genrepforest} for each forest in the collection $\mathcal{F}$. The properties (a)-(e) are easily verified. The property (f) holds for the union of pairs because the limited intersection of trees between different forests implies that a vertex can be an internal vertex in at most one forest.
\end{proof}

We will exhibit a cascading sequence of forests in $G\backslash M$ which essentially link up all the non matching color components.

\begin{lemma}\label{lem:int_disjoint_forests}
There exists a collection $\mathcal{F}$ of cascading sequence of forests in $G\backslash M$ such that:
\begin{enumerate}[{\rm (a)}]
\item $\sum_{F\in \mathcal{F}} |l(F)| = \sum_{i=1}^h (k_i - 1)$.
\item ${\rm Int}(\mathcal{F})\cap V(H)=\emptyset$ for all non-matching colour components $H$.
\item for all $F\in \mathcal{F}$ and for all $u\in r(F)\cup l(F)$, the colour of edges in $F$ incident at $u$ is the same as the matching colour at $u$.
\end{enumerate}
\end{lemma}
\begin{proof}
For each $m=1,\ldots,h$, we will exhibit a cascading sequence of forests $\mathcal{F}_m$ in $C_m$ satisfying $\sum_{F\in \mathcal{F}_m} |l(F)| = k_m - 1$ and ${\rm Int}(\mathcal{F}_m)\cap H=\emptyset$ for all non-matching color components $H$ contained in $C_m$. Then it is easily seen that $\mathcal{F} := \cup_{m=1}^h \mathcal{F}_m$ satisfies the requirements of the lemma. Let $C := C_m$ be an arbitrary but fixed component of $G\backslash M$. Let $H_1,\ldots, H_k$ (where $k = k_m$) be the non-matching color components contained in $C$. Let $V_H := \cup_{i=1}^k V(H_i)$. Let $F_1$ be a maximal forest with i) $r(F_1)\subseteq V(H_1)$, ii) $|F_1 \cap H_j| \leq 1$ for $j \in [k] \setminus \{1\}$ iii) $|l(F_1)|$ is maximal under this condition. Suppose forests $F_1,\ldots, F_j$ have been constructed such that they form a cascading sequence. Define $\mathcal{F'}_j := \{F_1,\ldots, F_j\}$, $I_j := \{i: V(\mathcal{F'}_j)\cap V(H_i)\neq \emptyset\}$ and let $K_j := [k]\backslash I_j$. Intuitively, $I_j$ is the indices of non-matching color components already visited by the collection ${\cal F}^{'}_j$, while $K_j$ is the indices which have not been visited so far. Let $F_{j+1}$ be the maximal forest in $C$ such that $r(F_{j+1})\subseteq \cup_{i\in I_j} V(H_i)$, $l(F_{j+1})\subseteq \cup_{i\in K_j} V(H_i)$ and ${\rm Int}(F_{j+1})\cap V_H = \emptyset$. Note that by connectedness of $C$, $F_{j+1}$ is non-empty for $I_j\neq [k]$. Thus $I_j \subset I_{j+1}$ unless $I_j = [k]$. Let $t$ be minimal such that $I_t = [k]$. Then $\mathcal{F}_m = \{F_1,\ldots, F_t\}$ is a cascading sequence of forests in $C$ with ${\rm Int}(\mathcal{F}_m)\cap V_H = \emptyset$. To complete the proof, we need to show that $\sum_{j=1}^t |l(F_i)| = k - 1$. This follows from the fact that for each $i\in \{2,\ldots,k\}$, there exists minimum $j$ such that $i\in I_j$. Then $l(F_j)\cap V(H_i)\neq \emptyset$. In other words, for each $i\in \{2,\ldots,k\}$, there exists a forest $F_j$ which has a leaf in the component $H_i$. The claim follows.   
\end{proof}

The above proof can be intuitively understood using Figure \ref{fig:noroot}, where the forest $F_1$ consists of two dashed trees, and the forest $F_2$ consists of the two solid trees. The collection $\mathcal{F}=\{F_1,F_2\}$ is the cascading sequence for the arrangement in Figure \ref{fig:noroot}.

\section{Repetition Content}
\label{sec:repetcontent}
In the previous section we exhibited pairs of vertices having the same matching colour incident at them. Intuitively, a lot of such pairs should imply certain amount of repetition of colours among matching edges. In this section we attempt to quantitatively estimate the repetition.

Call a set $S\subseteq V(G)$ $M$-\emph{monochromatic} if all the edges in $M$ incident with vertices in $S$ have the same colour. Let $i(S; M)$ denote the number of edges in $M$ which are incident with a vertex in $S$. For an $M$-monochromatic set $S$, let \emph{repetition content} of $S$, denoted by $\rp{S}$ be defined as $\rp{S} := i(S; M) - 1$. For $|S| = k$, we observe that $(k-2)/2 \leq \rp{S} \leq k-1$. The lower bound is attained when every vertex in $S$ has its matching partner also in $S$. The upper bound is obtained for sets where no edge in $M$ has both vertices in $S$. Following is an easy consequence of the definition of repetition content.

\begin{lemma}\label{lem:repcontent}
For $j\in \mathcal{C}_M$ (the set of matching colours in ${\cal C}$), let $S_1, S_2, \dots S_{|\mathcal{C}_M|}$ be a collection of $M$-monochromatic sets such that matching edges incident with $S_j$ are coloured $j$. Then $|\mathcal{C}_M|\leq |M|-\sum_{j\in \mathcal{C}_M} \rp{S_j}$.
\end{lemma}
 
\subsection{Repetition by matching colours}
We recall the notation from Section \ref{sec:keynotation}. Throughout the remainder of this section let $\mathcal{F}$ denote a cascading sequence of forests satisfying Lemma \ref{lem:int_disjoint_forests}. Let $V(\mathcal{F})$ denote the vertex set $\bigcup_{F\in \mathcal{F}} V(F)$. Then by Lemma \ref{lem:genrepforest_collection}, there exists set $P$ of repetition pairs of vertices $V(\mathcal{F})$ with $|P| = \sum_{i=1}^h (k_i-1)$, where $k_i$ is the number of non-matching colour components in component $C_i$ of $G\backslash M$.

We consider a partition of the pairs $P$ according to the incident matching colour. Let $P_j := \{(u,v)\in P: \mcl{u} = \mcl{v} = j\}$ for $j\in \mathcal{C}_M$. Let $C_j\subseteq P_j, j\in \mathcal{C}_M$ be the repetition pairs consisting of matching pairs, i.e. $C_j = \{(u,v)\in P_j: uv\in M\}$. Let $U_j := \Pi_1(P_j)$ and $V_j := \Pi_2(P_j)$ denote the projections of the set $P_j$ on first and second coordinate. Let $S_j := U_j\cup V_j$. We use $\rp{S_j}$ as the repetition associated with colour $j\in \mathcal{C}_M$. We call a colour $j\in \mathcal{C}_M$ as \emph{high} if $\rp{S_j}\geq (|P_j| - |C_j|)/2$. Otherwise, we call the colour as \emph{low}. Let $\mathcal{H}$ and $\mathcal{L}$ denote the high and low colours respectively.

\begin{lemma}\label{lem:repetitioncontent}
Let the sets $P, \mathcal{H}, \mathcal{L}$ and $P_j, C_j, U_j,V_j,S_j, j\in \mathcal{C}_M$ be as defined in this section. Then we have,
\begin{enumerate}[{\rm (a)}]
\item $|S_j|\geq |P_j|+1$ for all $j\in \mathcal{C}_M$. 
\item $\rp{S_j}\geq (|P_j| - 1)/2$ for all $j\in \mathcal{C}_M$.
\item If $C_j\neq \emptyset$ then $j\in \mathcal{H}$.
\item For $j\in \mathcal{L}$, all vertices in $S_j$ have their matching partner in $M$ also in $S_j$. Furthermore, $S_j$ has a unique maximum element $v^\ast$ with respect to ordering $\preceq_\mathcal{F}$.
\item $|S_j|\geq 4$ for $j\in {\cal L}$.
\end{enumerate}
\end{lemma}

\begin{proof}
From part (a) of Lemma \ref{lem:genrepforest_collection}, we have $|U_j|=|P_j|$. Let $v^\ast$ be a maximal element in $S_j$ under $\preceq_\mathcal{F}$. Clearly, by part (b) of Lemma \ref{lem:genrepforest_collection}, $v^\ast\not\in U_j$. Thus $|S_j|\geq |U_j|+1 = |P_j|+1$. This proves claim (a). Since $\rp{S}\geq (|S|-2)/2$ for all $S\subseteq V(G)$, it follows using claim (a) that $\rp{S_j}\geq (|P_j|-1)/2$ for all $j\in \mathcal{C}_M$. The claim (b) is thus proved. Claim (c) also follows from claim (b) and definition of $\mathcal{H}$. For claim (d), observe that for $j\in \mathcal{L}$, we have by claim (c), $|C_j|=0$, and hence $\rp{S_j}<|P_j|/2$ from the definition of low. Also by claim (b), we have $\rp{S_j}\geq (|P_j|-1)/2$. Since $\rp{S_j}$ is an integer, we conclude $\rp{S_j}=(|P_j|-1)/2$, and hence $|S_j|\leq |P_j|+1$. Since no element $U_j$ can be maximum, together with claim (a), we conclude $|S_j| = |P_j|+1$ and thus $S_j = U_j\cup \{v^\ast\}$. Clearly then $v^\ast$ is the unique maximal element in $S_j$. Also since $S_j$ satisfies $\rp{S_j} = (|S_j|-2)/2$, each element in $S_j$ must have its matching partner in $M$ in $S_j$. This proves claim (d). From claim (d) we see that $|S_j|\geq 4$ or $|S_j|=2$ for $j\in {\cal L}$. To prove claim (e), we rule out the possibility $|S_j|=2$. Note that $|S_j|=2$ implies $P_j=\{(u,v)\}$ where $uv\in M$. But then $C_j=\{(u,v)\}$, which contradicts the assumption $j\in {\cal L}$ by claim (c). This proves claim (e). 
\end{proof}

\begin{lemma}\label{lem:rctrianglefree}
If $G$ is triangle free and $j\in \mathcal{L}$ we have either (a) $|S_j|\geq 6$, or (b)  $|S_j|=4$ and there exist vertices $u,v$ in $S_j$ with $(u,v)\in P_j$ such that the path $u\mathcal{F}v$ has an interior vertex $z_{uv}$. 
\end{lemma}
\begin{proof}
From parts (d) and (e) of Lemma \ref{lem:repetitioncontent} we see that if $|S_j|<6$, we must have $|S_j|=4$. Let $v^\ast$ be the maximum element of $S_j$ according to the order $\preceq_\mathcal{F}$. Clearly $v^\ast\not\in U_j$. Let $u$ be the matching partner of $v^\ast$. Clearly $(u,v^\ast)\not\in P_j$ as $C_j=\emptyset$. Let $z$ be a vertex such that $(u,z)\in P_j$ and let $w$ be the matching partner of $z$. Recalling  $v^*$ is the unique maximum element of $S_j$, $S_j = U_j \cup \{v^*\}$. Thus $S_j=\{u,z,w,v^\ast\}$. We also must have $(z,v^\ast)\in P_j$. Since $(u,z)\in P_j$ and $(z,v^\ast)\in P_j$, by part (e) of Lemma \ref{lem:genrepforest_collection}, $w\not\in u{\cal F}z$ and $w\not\in z\mathcal{F}v^\ast$. Thus $w\not\in u\mathcal{F}v^\ast$. Since $G$ is triangle free, one of the paths $u\mathcal{F}z, z\mathcal{F}v^\ast$ must have an internal vertex, thus satisfying condition (b) of the lemma. 
\end{proof}

\begin{lemma}\label{lem:distinct_int_vertices}
Let $(u,v)\in P_i$ and $(z,w)\in P_j$ for some $i\neq j\in \mathcal{C}_M$. Then the paths $u\mathcal{F}v$ and $z\mathcal{F}w$ do not share an internal vertex.
\end{lemma}
\begin{proof}
If possible, let $x$ be an internal vertex of both $u\mathcal{F}v$ and $z\mathcal{F}w$. By part (e) of Lemma \ref{lem:genrepforest_collection}, $\mcl{x}\neq i$ and $\mcl{x}\neq j$. This is a contradiction, as $i$ and $j$ are the only two colors incident at $x$. Thus, the lemma is proved.
\end{proof}

Finally, we prove our main result:
\begin{proof}[Proof of Theorem \ref{thm:main}]
We have $|\mathcal{C}|=|\mathcal{C}_M|+|\mathcal{C}_N|$. From Lemma \ref{lem:repcontent}, we have,
\begin{equation}\label{eq:first}
|\mathcal{C}|\leq |M|-\sum_{j\in \mathcal{C}_M } \rp{S_j}+|\mathcal{C}_N|.
\end{equation}
Moreover from Lemma \ref{lem:repetitioncontent}(b) we have,
\begin{align}\label{eq:second}
\sum_{j\in \mathcal{C}_M} \rp{S_j} &= \sum_{j\in \mathcal{H}} \rp{S_j} + \sum_{j\in \mathcal{L}} \rp{S_j} \nonumber \\
& \geq \sum_{j\in \mathcal{H}} \frac{|P_j| - |C_j|}{2} + \sum_{j\in \mathcal{L}} \frac{|P_j|-1}{2} \nonumber \\
&= \sum_{j\in \mathcal{C}_M} \frac{|P_j|}{2} - \sum_{j\in \mathcal{H}}\frac{|C_j|}{2} - \frac{|\mathcal{L}|}{2}
\end{align}
Observing that $\sum_{j\in \mathcal{C}_M} |P_j| = \sum_{i=1}^h (k_i-1) = |\mathcal{C}_N| - h$, from \eqref{eq:first} and \eqref{eq:second}, we have,
\begin{align}\label{eq:third}
|\mathcal{C}| & \leq |\mathcal{C}_N| + |M| - \frac{|\mathcal{C}_N| - h}{2} + 
\frac{\sum_{j\in \mathcal{H}} |C_j|}{2} + \frac{|\mathcal{L}|}{2}
\end{align}
Let $\Delta := \sum_{j\in \mathcal{H}} |C_j|$. Note that the path $u\mathcal{F}v$ contains at least one internal vertex for $(u,v)\in \bigcup_{j\in \mathcal{H}}
C_j$. From Lemma \ref{lem:distinct_int_vertices} and part (f) of Lemma \ref{lem:genrepforest_collection} it follows that all these internal vertices are distinct. Thus, the collection $\mathcal{F}$ has at least $\Delta$ internal vertices. Since $\mathcal{F}$ was chosen according to Lemma \ref{lem:int_disjoint_forests} and each non-matching color component contains at least two vertices, we have:
\begin{equation}\label{eq:fourth}
2|\mathcal{C}_N| + \Delta \leq n.
\end{equation}
Substituting, $|\mathcal{C}_N|\leq (n-\Delta)/2$ in \eqref{eq:third}, we have:
\begin{equation}\label{eq:fifth}
|\mathcal{C}| \leq \frac{3|M|}{2} + \frac{\Delta + 2|{\cal L}|}{4} + \frac{h}{2}
\end{equation}
Since each $j\in \mathcal{L}$ determines at least $2$ matching edges of the same color we have $|\mathcal{C}_M|\leq |M| - |\mathcal{L}|$. Combining with
\eqref{eq:fourth} we have:
\begin{equation}\label{eq:sixth}
|\mathcal{C}|\leq 2|M| - \frac{\Delta + 2|\mathcal{L}|}{2}.
\end{equation}
Equating the upperbounds on $|\mathcal{C}|$ in \eqref{eq:fifth} and \eqref{eq:sixth}, we obtain $|\mathcal{C}|\leq \frac{5}{3}(|M|+h)$. This completes the proof of part (a) of the theorem.

Now, consider the case when $G$ is additionally {\em triangle free}. In this case we write ${\cal L}={\cal L}_1\uplus {\cal L}_2$ where ${\cal L}_1 = \{j\in
{\cal L}: |S_j|\geq 6\}$ and ${\cal L}_2 = \{j\in {\cal L}: |S_j|=4\}$. Since each $j\in \mathcal{L}_1$ determines at least $3$ matching edges of the same color, and each $j\in \mathcal{L}_2$ determines at least $2$ matching edges of the same color, we conclude that $|\mathcal{C}_M|\leq |M| - 2|\mathcal{L}_1| - |\mathcal{L}_2|$. Together with trivial bound $|\mathcal{C}_N|\leq |M|$, we have:
\begin{equation}\label{eq:seventh}
|\mathcal{C}| \leq 2|M| - 2|\mathcal{L}_1| - |\mathcal{L}_2|.
\end{equation}
Also, the equivalent of Equation \eqref{eq:third} for this case can be written as:
\begin{align}\label{eq:eight}
|\mathcal{C}| & \leq |\mathcal{C}_N| + |M| - \frac{|\mathcal{C}_N| - h}{2} + \frac{\sum_{j\in \mathcal{H}} |C_j|}{2} + \frac{|\mathcal{L}_1| + |\mathcal{L}_2|}{2}
\end{align}
Finally, we exploit the triangle free property of the graph to account for more internal vertices. Observing that each $u{\cal F}v$ path now has at least two internal vertices for $(u,v)\in \bigcup_{j\in {\cal H}} C_j$, and further by Lemma \ref{lem:rctrianglefree}, for every $j\in \mathcal{L}_2$ there is a vertex $z_j$ which is an internal vertex of $\mathcal{F}$. As before, these internal vertices are distinct. Thus we have at least $2\Delta + |{\cal L}_2|$ internal vertices in ${\cal F}$. Thus we have:
\begin{equation}\label{eq:nine}
2|\mathcal{C}_N| + 2\Delta + |\mathcal{L}_2|\leq n.
\end{equation}
Substituting $|\mathcal{C}_N|\leq (n-2\Delta - |\mathcal{L}_2|)/2$ in \eqref{eq:eight} we get (recall $\Delta := \sum_{j\in {\cal H}} |C_j|$),
\begin{equation}\label{eq:ten}
|\mathcal{C}| \leq \frac{3|M|}{2} + \frac{2|\mathcal{L}_1|+ |\mathcal{L}_2|}{4} + \frac{h}{2}
\end{equation}
 From Equations \eqref{eq:seventh} and \eqref{eq:ten}, by equating the upper bounds, we get $|\mathcal{C}|\leq \frac{8}{5}(|M|+h)$, which proves the  part (b) of the theorem.
\end{proof}

\section{A Lower Bound for the Matching Based Algorithm}
In this section, we show a triangle free graph with perfect matching establishing a lower bound of $58/37$ on the approximability of Algorithm \ref{algo01} on such graphs. The factor $58/37$ sits roughly midway between the known lower bound of $3/2$ on the approximability of max edge $2$-coloring, and the
approximation factor of $8/5$ shown for Algorithm \ref{algo01} in this paper.

Let us consider the graph $G$ on $72$ vertices as shown in the Figure~\ref{fig:lowerbound}. Clearly, the graph has a perfect matching which is marked with the bold lines. If we delete the matching edges we get one connected component. So the Algorithm \ref{algo01} outputs $37$ colours. On the other hand if we colour the graph by below mentioned colouring scheme, we can use $58$ colours to colour the entire graph. 

\noindent \textbf{Colouring Scheme:} For every matched edges that are represented with a bold straight line in the Figure \ref{fig:lowerbound} and two adjacent edges on its either sides in the drawing, use different colours to colour them. For the top six edges represented with vertical straight lines in the drawing, colour them with different colours. Now delete all the coloured edges. The remaining graph has seven connected components. Colour all edges in a connected component with same colours.

With the above colouring scheme we use total of $3 \times 15 + 6 + 7 = 58$ colours. It is easy to verify that it is a valid colouring scheme.
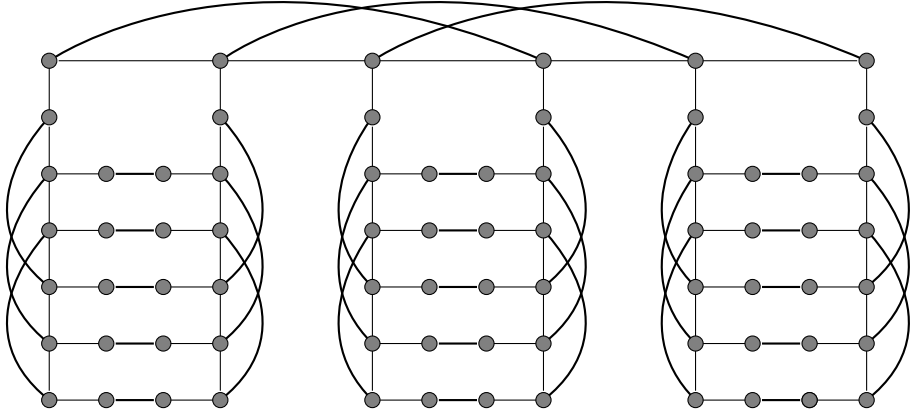
\begin{figure}[h]
\centering
\usetikzlibrary{decorations.pathmorphing}
\begin{tikzpicture}[scale = 0.5]

\draw (-2.5,3) node (v3) {} -- (-2.5,4.5) node (v1) {};
\draw (2,3) node (v5) {} -- (2,4.5);
\draw (-2.5,1.5) -- (-1,1.5) node (v15) {};
\draw (0.5,1.5) node (v16) {} -- (2,1.5);
\draw (-2.5,0) -- (-1,0) node (v17) {};
\draw (0.5,0) node (v18) {} -- (2,0);
\draw (-2.5,-1.5) -- (-1,-1.5) node (v19) {};
\draw (0.5,-1.5) node (v20) {} -- (2,-1.5);
\draw (-2.5,-3) -- (-1,-3) node (v23) {};
\draw (0.5,-3) node (v24) {} -- (2,-3);
\draw (-2.5,-4.5) node (v4) {} -- (-1,-4.5) node (v21) {};
\draw (0.5,-4.5) node (v22) {} -- (2,-4.5) node (v6) {};
\draw [thick](-2.5,3) .. controls (-2.5,3) and (-5,0.5) .. (-2.5,-1.5);
\draw [thick](-2.5,1.5) .. controls (-2.5,1.5) and (-5,-1) .. (-2.5,-3);
\draw [thick](-2.5,0) .. controls (-2.5,0) and (-5,-2.5) .. (-2.5,-4.5);
\draw [thick](2,3) .. controls (2,3) and (4.5,0.5) .. (2,-1.5);
\draw [thick](2,1.5) .. controls (2,1.5) and (4.5,-1) .. (2,-3);
\draw [thick](2,0) .. controls (2,0) and (4.5,-2.5) .. (2,-4.5);

\draw (6,4.5) -- (6,3) node (v7) {};
\draw (10.5,4.5) -- (10.5,3) node (v9) {};
\draw (6,1.5) -- (7.5,1.5) node (v25) {};
\draw (6,0) -- (7.5,0) node (v27) {};
\draw (6,-1.5) -- (7.5,-1.5) node (v29) {};
\draw (6,-3) -- (7.5,-3) node (v31) {};
\draw (6,-4.5) node (v8) {} -- (7.5,-4.5) node (v33) {};
\draw (9,1.5) node (v26) {} -- (10.5,1.5);
\draw (9,0) node (v28) {} -- (10.5,0);
\draw (9,-1.5) node (v30) {} -- (10.5,-1.5);
\draw (9,-3) node (v32) {} -- (10.5,-3);
\draw (9,-4.5) node (v34) {} -- (10.5,-4.5) node (v10) {};
\draw [thick](6,3) .. controls (6,3) and (4,0.5) .. (6,-1.5);
\draw [thick](6,1.5) .. controls (6,1.5) and (4,-1) .. (6,-3);
\draw [thick](6,0) .. controls (6,0) and (4,-2.5) .. (6,-4.5);
\draw [thick](10.5,3) .. controls (10.5,3) and (13,0.5) .. (10.5,-1.5);
\draw [thick](10.5,1.5) .. controls (10.5,1.5) and (13,-1) .. (10.5,-3);
\draw [thick](10.5,0) .. controls (10.5,0) and (13,-2.5) .. (10.5,-4.5);

\draw (14.5,0) -- (16,0) node (v37) {};
\draw (14.5,1.5) -- (16,1.5) node (v35) {};
\draw (14.5,-1.5) -- (16,-1.5) node (v39) {};
\draw (14.5,-3) -- (16,-3) node (v41) {};
\draw (14.5,-4.5) node (v12) {} -- (16,-4.5) node (v43) {};
\draw (17.5,1.5) node (v36) {} -- (19,1.5);
\draw (17.5,0) node (v38) {} -- (19,0);
\draw (17.5,-1.5) node (v40) {} -- (19,-1.5);
\draw (17.5,-3) node (v42) {} -- (19,-3);
\draw (17.5,-4.5) node (v44) {} -- (19,-4.5) node (v14) {};
\draw (14.5,4.5) -- (14.5,3) node (v11) {};
\draw (19,4.5) node (v2) {} -- (19,3) node (v13) {};
\draw [thick](14.5,3) .. controls (14.5,3) and (12.5,0.5) .. (14.5,-1.5);
\draw [thick](14.5,1.5) .. controls (14.5,1.5) and (12.5,-1) .. (14.5,-3);
\draw [thick](14.5,0) .. controls (14.5,0) and (12.5,-2.5) .. (14.5,-4.5);
\draw [thick](19,3) .. controls (19,3) and (21.5,0.5) .. (19,-1.5);
\draw [thick](19,1.5) .. controls (19,1.5) and (21.5,-1) .. (19,-3);
\draw [thick](19,0) .. controls (19,0) and (21.5,-2.5) .. (19,-4.5);
\draw [thick](-2.5,4.5) .. controls (-2.5,4.5) and (2.5,8) .. (10.5,4.5);
\draw [thick](2,4.5) .. controls (2,4.5) and (6.5,8) .. (14.5,4.5);
\draw [thick](6,4.5) .. controls (6,4.5) and (11,8) .. (19,4.5);

\draw (v1) -- (v2);
\draw (v3) -- (v4);
\draw (v5) -- (v6);
\draw (v7) -- (v8);
\draw (v9) -- (v10);
\draw (v11) -- (v12);
\draw (v13) -- (v14);

\draw [thick](v15) -- (v16);
\draw [thick](v17) -- (v18);
\draw [thick](v19) -- (v20);
\draw [thick](v21) -- (v22);
\draw [thick](v23) -- (v24);
\draw [thick](v25) -- (v26);
\draw [thick](v27) -- (v28);
\draw [thick](v29) -- (v30);
\draw [thick](v31) -- (v32);
\draw [thick](v33) -- (v34);
\draw [thick](v35) -- (v36);
\draw [thick](v37) -- (v38);
\draw [thick](v39) -- (v40);
\draw [thick](v41) -- (v42);
\draw [thick](v43) -- (v44);

\draw [fill=gray] (v1) ellipse (0.2 and 0.2);
\draw [fill=gray] (v2) ellipse (0.2 and 0.2);
\draw [fill=gray] (v3) ellipse (0.2 and 0.2);
\draw [fill=gray] (v4) ellipse (0.2 and 0.2);
\draw [fill=gray] (v5) ellipse (0.2 and 0.2);
\draw [fill=gray] (v6) ellipse (0.2 and 0.2);
\draw [fill=gray] (v7) ellipse (0.2 and 0.2);
\draw [fill=gray] (v8) ellipse (0.2 and 0.2);
\draw [fill=gray] (v9) ellipse (0.2 and 0.2);
\draw [fill=gray] (v10) ellipse (0.2 and 0.2);
\draw [fill=gray] (v11) ellipse (0.2 and 0.2);
\draw [fill=gray] (v12) ellipse (0.2 and 0.2);
\draw [fill=gray] (v13) ellipse (0.2 and 0.2);
\draw [fill=gray] (v14) ellipse (0.2 and 0.2);
\draw [fill=gray] (v15) ellipse (0.2 and 0.2);
\draw [fill=gray] (v16) ellipse (0.2 and 0.2);
\draw [fill=gray] (v17) ellipse (0.2 and 0.2);
\draw [fill=gray] (v18) ellipse (0.2 and 0.2);
\draw [fill=gray] (v19) ellipse (0.2 and 0.2);
\draw [fill=gray] (v20) ellipse (0.2 and 0.2);
\draw [fill=gray] (v21) ellipse (0.2 and 0.2);
\draw [fill=gray] (v22) ellipse (0.2 and 0.2);
\draw [fill=gray] (v23) ellipse (0.2 and 0.2);
\draw [fill=gray] (v24) ellipse (0.2 and 0.2);
\draw [fill=gray] (v25) ellipse (0.2 and 0.2);
\draw [fill=gray] (v26) ellipse (0.2 and 0.2);
\draw [fill=gray] (v27) ellipse (0.2 and 0.2);
\draw [fill=gray] (v28) ellipse (0.2 and 0.2);
\draw [fill=gray] (v29) ellipse (0.2 and 0.2);
\draw [fill=gray] (v30) ellipse (0.2 and 0.2);
\draw [fill=gray] (v31) ellipse (0.2 and 0.2);
\draw [fill=gray] (v32) ellipse (0.2 and 0.2);
\draw [fill=gray] (v33) ellipse (0.2 and 0.2);
\draw [fill=gray] (v34) ellipse (0.2 and 0.2);
\draw [fill=gray] (v35) ellipse (0.2 and 0.2);
\draw [fill=gray] (v36) ellipse (0.2 and 0.2);
\draw [fill=gray] (v37) ellipse (0.2 and 0.2);
\draw [fill=gray] (v38) ellipse (0.2 and 0.2);
\draw [fill=gray] (v39) ellipse (0.2 and 0.2);
\draw [fill=gray] (v40) ellipse (0.2 and 0.2);
\draw [fill=gray] (v41) ellipse (0.2 and 0.2);
\draw [fill=gray] (v42) ellipse (0.2 and 0.2);
\draw [fill=gray] (v43) ellipse (0.2 and 0.2);
\draw [fill=gray] (v44) ellipse (0.2 and 0.2);

\draw [fill=gray] (2,4.5) ellipse (0.2 and 0.2);
\draw [fill=gray] (6,4.5) ellipse (0.2 and 0.2);
\draw [fill=gray] (10.5,4.5) ellipse (0.2 and 0.2);
\draw [fill=gray] (14.5,4.5) ellipse (0.2 and 0.2);
\draw [fill=gray] (-2.5,1.5) ellipse (0.2 and 0.2);
\draw [fill=gray] (-2.5,0) ellipse (0.2 and 0.2);
\draw [fill=gray] (-2.5,-1.5) ellipse (0.2 and 0.2);
\draw [fill=gray] (-2.5,-3) ellipse (0.2 and 0.2);
\draw [fill=gray] (2,1.5) ellipse (0.2 and 0.2);
\draw [fill=gray] (2,0) ellipse (0.2 and 0.2);
\draw [fill=gray] (2,-1.5) ellipse (0.2 and 0.2);
\draw [fill=gray] (2,-3) ellipse (0.2 and 0.2);
\draw [fill=gray] (6,1.5) ellipse (0.2 and 0.2);
\draw [fill=gray] (6,0) ellipse (0.2 and 0.2);
\draw [fill=gray] (6,-1.5) ellipse (0.2 and 0.2);
\draw [fill=gray] (6,-3) ellipse (0.2 and 0.2);
\draw [fill=gray] (10.5,1.5) ellipse (0.2 and 0.2);
\draw [fill=gray] (10.5,0) ellipse (0.2 and 0.2);
\draw [fill=gray] (10.5,-1.5) ellipse (0.2 and 0.2);
\draw [fill=gray] (10.5,-3) ellipse (0.2 and 0.2);
\draw [fill=gray] (14.5,1.5) ellipse (0.2 and 0.2);
\draw [fill=gray] (14.5,0) ellipse (0.2 and 0.2);
\draw [fill=gray] (14.5,-1.5) ellipse (0.2 and 0.2);
\draw [fill=gray] (14.5,-3) ellipse (0.2 and 0.2);
\draw [fill=gray] (19,1.5) ellipse (0.2 and 0.2);
\draw [fill=gray] (19,0) ellipse (0.2 and 0.2);
\draw [fill=gray] (19,-1.5) ellipse (0.2 and 0.2);
\draw [fill=gray] (19,-3) ellipse (0.2 and 0.2);
\draw [fill=gray] (v44) ellipse (0.2 and 0.2);

\end{tikzpicture}
\caption{The Algorithm \ref{algo01} exhibits an approximation factor of $58/37$ when runs on the above graph that has a perfect matching. The matching edges are shown in bold lines.}
\label{fig:lowerbound}
\end{figure}

\bibliography{presubmission_else}

\begin{thebibliography}{MSTV96}

\bibitem[AJK04]{AxenoJK04}
Maria Axenovich, Tao Jiang, and Andr{\'{e}} K{\"{u}}ndgen.
\newblock Bipartite anti-{R}amsey numbers of cycles.
\newblock {\em Journal of Graph Theory}, 47(1):9--28, 2004.

\bibitem[Alo83]{Alon83}
Noga Alon.
\newblock On a conjecture of {E}rd{\"{o}}s, {S}imonovits, and {S}{\'{o}}s
  concerning anti-{R}amsey theorems.
\newblock {\em Journal of Graph Theory}, 7(1):91--94, 1983.

\bibitem[AP16]{AdamP}
Anna Adamaszek and Alexandru Popa.
\newblock Approximation and hardness results for the maximum edge q-coloring
  problem.
\newblock {\em Journal of Discrete Algorithms}, 38-41:1--8, 2016.

\bibitem[BFGR01]{BlokhuFGR01}
Aart Blokhuis, Ralph~J. Faudree, Andr{\'{a}}s Gy{\'{a}}rf{\'{a}}s, and
  Mikl{\'{o}}s Ruszink{\'{o}}.
\newblock Anti-{R}amsey colorings in several rounds.
\newblock {\em Journal of Combinatorial Theory, Ser. {B}}, 82(1):1--18, 2001.

\bibitem[CMRS18]{Chandran}
L.~Sunil Chandran, Rogers Mathew, Deepak Rajendraprasad, and Nitin Singh.
\newblock Approximation bounds on maximum edge 2-coloring of dense graphs.
\newblock {\em CoRR}, abs/1810.00624, 2018.

\bibitem[ESS75]{Erdos}
Paul {E}rd{\"{o}}s, Mikl{\'{o}}s {S}imonovits, and Vera~{T.} {S}{\'{o}}s.
\newblock Anti-{R}amsey theorems.
\newblock {\em Infinite and finite sets (Colloquium, Keszthely, 1973; dedicated
  to P. Erd{\"{o}}s on his 60th birthday)}, 10(II):633--643, 1975.

\bibitem[FKSS09]{FujitaKSS09}
Shinya Fujita, Atsushi Kaneko, Ingo Schiermeyer, and Kazuhiro Suzuki.
\newblock A rainbow k-matching in the complete graph with r colors.
\newblock {\em Electronic Journal of Combinatorics}, 16(1), 2009.

\bibitem[FMO10]{FujitaMO}
Shinya Fujita, Colton Magnant, and Kenta Ozeki.
\newblock Rainbow generalizations of {R}amsey theory: {A} survey.
\newblock {\em Graphs and Combinatorics}, 26(1):1--30, 2010.

\bibitem[FZW09]{Feng}
Wangsen Feng, Li'ang Zhang, and Hanpin Wang.
\newblock Approximation algorithm for maximum edge coloring.
\newblock {\em Theoretical Computer Science}, 410(11):1022--1029, 2009.

\bibitem[GKM13]{Goyal}
Prachi Goyal, Vikram Kamat, and Neeldhara Misra.
\newblock On the parameterized complexity of the maximum edge 2-coloring
  problem.
\newblock In {\em Mathematical Foundations of Computer Science 2013 - 38th
  International Symposium, {MFCS} 2013, Klosterneuburg, Austria, August 26-30,
  2013.}, pages 492--503, 2013.

\bibitem[GL10]{GorgolL10}
Izolda Gorgol and Ewa Lazuka.
\newblock Rainbow numbers for small stars with one edge added.
\newblock {\em Discussiones Mathematicae Graph Theory}, 30(4):555--562, 2010.

\bibitem[HY12]{HaasY12}
Ruth Haas and Michael Young.
\newblock The anti-{R}amsey number of perfect matching.
\newblock {\em Discrete Mathematics}, 312(5):933--937, 2012.

\bibitem[Jia02a]{JiangJCT02}
Tao Jiang.
\newblock Anti-{R}amsey numbers of subdivided graphs.
\newblock {\em Journal of Combinatorial Theory, Ser. {B}}, 85(2):361--366,
  2002.

\bibitem[Jia02b]{Jiang02}
Tao Jiang.
\newblock Edge-colorings with no large polychromatic stars.
\newblock {\em Graphs and Combinatorics}, 18(2):303--308, 2002.

\bibitem[JW03]{JiangW03}
Tao Jiang and Douglas~B. West.
\newblock On the {E}rd{\"{o}}s-{S}imonovits-{S}{\'{o}}s conjecture about the
  anti-{R}amsey number of a cycle.
\newblock {\em Combinatorics, Probability {\&} Computing}, 12(5-6):585--598,
  2003.

\bibitem[JW04]{JiangW04}
Tao Jiang and Douglas~B. West.
\newblock Edge-colorings of complete graphs that avoid polychromatic trees.
\newblock {\em Discrete Mathematics}, 274(1-3):137--145, 2004.

\bibitem[MN00]{MontelN00}
Juan~Jos{\'{e}} Montellano{-}Ballesteros and Victor Neumann{-}Lara.
\newblock Totally multicoloured cycles.
\newblock {\em Electronic Notes in Discrete Mathematics}, 5:239--242, 2000.

\bibitem[MN02]{MontelN02}
Juan~Jos{\'{e}} Montellano{-}Ballesteros and Victor Neumann{-}Lara.
\newblock An anti-{R}amsey theorem.
\newblock {\em Combinatorica}, 22(3):445--449, 2002.

\bibitem[MN03]{MontelN03}
Juan~Jos{\'{e}} Montellano{-}Ballesteros and Victor Neumann{-}Lara.
\newblock A linear heterochromatic number of graphs.
\newblock {\em Graphs and Combinatorics}, 19(4):533--536, 2003.

\bibitem[MN05]{MontelN05}
Juan~Jos{\'{e}} Montellano{-}Ballesteros and Victor Neumann{-}Lara.
\newblock An anti-{R}amsey theorem on cycles.
\newblock {\em Graphs and Combinatorics}, 21(3):343--354, 2005.

\bibitem[Mon06]{Montel06}
Juan~Jos{\'{e}} Montellano{-}Ballesteros.
\newblock On totally multicolored stars.
\newblock {\em Journal of Graph Theory}, 51(3):225--243, 2006.

\bibitem[MSTV96]{ManousSTV96}
Yannis Manoussakis, M.~Spyratos, Zsolt Tuza, and Margit Voigt.
\newblock Minimal colorings for properly colored subgraphs.
\newblock {\em Graphs and Combinatorics}, 12(1):345--360, 1996.

\bibitem[RC05]{Raniwala}
Ashish Raniwala and Tzi{-}cker Chiueh.
\newblock Architecture and algorithms for an {IEEE} 802.11-based multi-channel
  wireless mesh network.
\newblock In {\em {INFOCOM} 2005. 24th Annual Joint Conference of the {IEEE}
  Computer and Communications Societies, 13-17 March 2005, Miami, FL, {USA}},
  pages 2223--2234, 2005.

\bibitem[Sch04]{Schier04}
Ingo Schiermeyer.
\newblock Rainbow numbers for matchings and complete graphs.
\newblock {\em Discrete Mathematics}, 286(1-2):157--162, 2004.

\bibitem[SS84]{Simono}
Mikl{\'{o}}s Simonovits and {Vera T.} S{\'o}s.
\newblock On restricted colourings of {$K_n$}.
\newblock {\em Combinatorica}, 4(1):101--110, 1984.

\end{thebibliography}
\bibliographystyle{alpha}
\nocite{*}
\end{document}